\documentclass[11pt]{article}
\usepackage{graphicx,amsmath,amsfonts,amssymb,dcolumn,amsthm}

\newtheorem{thma}{Condition}[section]
\newtheorem{cor}[thma]{Requirement}

\newtheorem{stmt}[thma]{Statement}

\begin{document}

\title{\textbf{On the elimination of infinitesimal Gribov ambiguities in non-Abelian gauge theories}}
\author{\textbf{Ant\^onio D.~Pereira Jr.}\thanks{aduarte@if.uff.br} \  and \  \textbf{Rodrigo F.~Sobreiro}\thanks{sobreiro@if.uff.br}\\\\
\textit{{\small UFF $-$ Universidade Federal Fluminense,}}\\
\textit{{\small Instituto de F\'{\i}sica, Campus da Praia Vermelha,}}\\
\textit{{\small Avenida General Milton Tavares de Souza s/n, 24210-346,}}\\
\textit{{\small Niter\'oi, RJ, Brasil.}}}
\date{}
\maketitle

\begin{abstract}
\noindent An alternative method to account for the Gribov ambiguities in gauge theories is presented. It is shown that, to eliminate Gribov ambiguities, at infinitesimal level, it is required to break the BRST symmetry in a soft manner. This can be done by introducing a suitable extra constraint that eliminates the infinitesimal Gribov copies. It is shown that the present approach is consistent with the well established known cases in the literature, \emph{i.e.}, the Landau and maximal Abelian gauges. The method is valid for gauges depending exclusively on the gauge field and is restricted to classical level. However, occasionally, we deal with quantum aspects of the technique, which are used to improve the results.
\end{abstract}

\section{Introduction}\label{INTRO}

It was shown in the seminal work by V.~N.~Gribov \cite{Gribov:1977wm} that non-Abelian gauge theories have a residual gauge symmetry that survives the standard Faddeev-Popov gauge fixing procedure \cite{Faddeev:1967fc}. This means that the gauge fixed path integral still carries redundant configurations contributing to the probabilities of physical processes. This problem is commonly known as Gribov ambiguities and the spurious configurations are called Gribov copies. Remarkably, the Gribov problem is only relevant at low energies, being negligible at the ultraviolet sector. Also in \cite{Gribov:1977wm}, it was argued that, to eliminate these ambiguities at the Landau gauge, one should truncate the range of integration in the path integral by introducing a suitable nonlocal term associated with the ghost propagator, see also \cite{Sobreiro:2005ec}. In fact, it was shown that infinitesimal Gribov copies are associated with the zero modes of the Faddeev-Popov operator, or equivalently, to the poles of the ghost propagator \cite{Gribov:1977wm,Sobreiro:2005ec}. The respective truncated domain of integration in the functional space is called the first Gribov region and is achieved by the introduction of a suitable non-local term associated with the ghost propagator. This region is not entirely free from copies, however, at least the infinitesimal ones are eliminated. Moreover, all configurations outside the first Gribov horizon are copies of configurations inside the horizon, see for instance \cite{Dell'Antonio:1991xt}. The result is that the propagators of the theory are dramatically modified. The gluon propagator acquires imaginary poles and is suppressed at the infrared regime. The pole is identified with a mass parameter called Gribov parameter which is determined by a self-consistent gap equation. The ghost propagator on the other hand is found to be enhanced with a behavior of $\sim1/k^4$, where $k$ is the momentum. These results are often interpreted as confinement evidence. It is important to recall that the Gribov ambiguities are inherent to the topology of non-Abelian gauge theories, and does not depend on the gauge choice \cite{Singer:1978dk} in order to infect the perturbative path integral. Still in \cite{Singer:1978dk}, it was shown that the main problem lies on the fact that Yang-Mills theories are formally constructed over an infinite-dimensional nontrivial principal bundle \cite{Kobayashi,Daniel:1979ez,CottaRamusino:1985ad,Falqui:1985iu,Nakahara:1990th,Bertlmann:1996xk} which implies that no global section can be defined and, thus, no global gauge fixing is possible.

Although present at any gauge choice, the treatment of this problem depends on the gauge choice, resulting in different effects for each treatable gauge fixing. In fact, until now, only two renormalizable gauges are known to be manageable within the Gribov ambiguities elimination, namely, the Landau \cite{Gribov:1977wm,Sobreiro:2005ec,Zwanziger:1989mf,Dell'Antonio:1989jn,Dell'Antonio:1991xt,Zwanziger:1992qr,Maggiore:1993wq,Dudal:2005na,Dudal:2008sp,Dudal:2011gd} and the maximal Abelian gauges \cite{Capri:2005tj,Capri:2006cz,Capri:2008ak,Capri:2008vk,Capri:2010an}. In particular, local and renormalizable actions that take into account the elimination of a considerable amount of Gribov copies were developed \cite{Zwanziger:1989mf,Maggiore:1993wq,Capri:2006cz}. This type of actions are commonly known as Gribov-Zwanziger actions. Moreover, extra effects such as the condensation of local composite operators were also taken into account within the Gribov-Zwanziger framework \cite{Dudal:2005na,Capri:2008ak}. In particular, the so-called refined Gribov-Zwanziger approach \cite{Dudal:2008sp,Dudal:2011gd}, where several condensates are taken into account at the Landau gauge, has remarkable agreement with lattice numerical simulations \cite{Bogolubsky:2007ud,Cucchieri:2007md,Sternbeck:2007ug,Cucchieri:2009zt}. In a refined treatment the gluon propagator also acquires complex poles, but is finite at zero momentum, while the ghost propagator maintains his perturbative typical behavior $\sim1/k^2$.

In the case of the Landau gauge, the Gribov ambiguities treatment was related to a BRST soft symmetry breaking \cite{Baulieu:2008fy,Baulieu:2009xr}. It was shown that the effects of the Gribov ambiguities elimination can be formally understood as a BRST breaking term which is proportional to a mass parameter (the Gribov parameter \cite{Gribov:1977wm}). The presence of the mass in front of a quadratic term on mixed fields is responsible to affect only the infrared sector. This is a good feature because the ultraviolet regime is preserved, as it should be.

In this work we develop a new method to eliminate Gribov infinitesimal copies in any gauge that depends exclusively on the gauge field, \emph{i.e.}, the gauge condition can be written as $\Delta(A)=0$. Essentially, the method is to implement an extra constraint that eliminates the zero modes of the Faddeev-Popov operator, \emph{i.e.}, the constraint ruins the Gribov copies equation. The technique is based on a homotopy between infinitesimal gauge transformations and BRST transformations and on the fact that the Faddeev-Popov operator is directly related to the Gribov copies. In fact, since Gribov copies equation is obtained from the invariance of the gauge fixing under infinitesimal gauge transformations, and due to the fact that the homotopic relation above mentioned has the same form for BRST or gauge transformations with respect to the field $A_\mu^A$, the copies equation can be derived from the BRST invariance of the gauge fixing. Since the copies equation is a zero mode equation, to eliminate the ambiguities we demand that no zero modes can be allowed. It is also shown that, to eliminate the zero modes, a BRST symmetry breaking is required. It is then argued that this breaking must be soft in order to preserve the ultraviolet sector. To do so, we develop a set of relatively simple rules which ensures that infinitesimal copies are eliminated from the path integral. Although different from the (refined) Gribov-Zwanziger approach and the BRST soft breaking mechanism, the present method has some similarities with both methods. In particular, a set of auxiliary fields are introduced and are recognized as the Gribov-Zwanziger fields. Also, in order to preserve the ultraviolet limit and renormalizability, BRST soft symmetry breaking arguments are used. Even though a few quantum arguments are used, the method remains at classical level and the resulting action is classical and free of infinitesimal Gribov copies. The final improved action must be determined by renormalizability arguments within the algebraic renormalization framework \cite{Piguet:1995er}.

It is worth mention that, in the usual Gribov-Zwanziger approaches, the hermiticity of the Faddeev-Popov operator is essential to the elimination of Gribov ambiguities. From the study of its eigenvalues, it is possible to establish a geometrical meaning of the first Gribov region in the gluon configuration space, see for instance \cite{Sobreiro:2005ec}. In the present case, instead of this geometric appeal, we propose a direct elimination of the ambiguities, by ruining the copies equation. Therefore, hermiticity does not play a relevant role. This enlarges the applicability of the method. As a consistence check, we apply the method to the Landau and maximal Abelian gauges and recover the usual results. 

Another interesting feature that the present method implies is that the gap equation that dynamically fixes the Gribov parameter can be modified. Because of the introduction of the mass parameter, the general action is allowed to carry extra massive terms. Thus, the refined term is already present, just like any other term allowed by the Ward identities of the specific gauge choice (for instance, at the Landau gauge, a massive term $A^A_\mu A^A_\mu$ could be present). The difference is that all of these terms are proportional to the Gribov mass parameter, which is the reason why the gap equation is modified. If we think of the Landau gauge, or the maximal Abelian gauge, the usual gap equation actually throws the theory right at the first Gribov horizon. It happens that this is precisely the place where the infinitesimal copies actually live. As suggested in the refined approach \cite{Vandersickel:2011zc} and lattice results \cite{Cucchieri:2011ig}, a deformation of the horizon would avoid such apparent inconsistency. This is exactly what the modified gap equation does here, however, from the very beginning of the construction instead of from dynamical effects.

Finally, we would like to point out that, recently, an interesting alternative method to account for the Gribov ambiguities has also been developed \cite{Serreau:2012cg}. In this work, the authors claim that a renormalizable action free of Landau pole is obtained and the propagators are also mass dependent. Essentially, instead of eliminating the Gribov ambiguities, it is taken an average over them with a suitable weight, just as suggested by Gribov at his original work \cite{Gribov:1977wm}. Another alternative technique was developed in \cite{vanBaal:1991zw}, in Hamiltonian formalism. It is also worth mention some recent work \cite{Slavnov:2008xz,Quadri:2010vt} that also goes in this direction, where the authors argue that a class of algebraic gauges are free of Gribov copies, leading to a renormalizable action. Essentially, the reason for these gauges do not develop Gribov copies is that the gauge condition does not involve a differential operator. In order to deal with Lorentz invariance in algebraic gauges, a modified formulation of Yang-Mills theory is necessary.

This article is organized as follows: In Sect.~\ref{smmry} we provide a summary of the rules that have to be followed in order to eliminate the infinitesimal Gribov ambiguities from a gauge theory. In Sect.~\ref{GENERAL} we introduce the basic ideas concerning Yang-Mills theories, BRST quantization and the geometrical aspects of non-Abelian gauge theories. Then, in Sect.~\ref{GRIBOV} we briefly review the Gribov problem and state three important properties that are crucial for the method. The method itself is discussed and developed in details in Sect.~\ref{SBRSTSB}. In Sect.~\ref{Qaspects} we discuss some quantum aspects of the model through universal Ward identities and how we can achieve the well-known gap equation with this approach. This method allows the introduction of terms which generalize the usual gap equation and this is constructed in Sect.~\ref{Gengape}. Finally, in Sect.~\ref{tests} we apply the method to two well established cases for consistency checks: the Landau and maximal Abelian gauges.

\section{Summary of rules to eliminate infinitesimal Gribov ambiguities}\label{smmry}

Before we begin with formal discussions, since we will develop a method to eliminate infinitesimal Gribov copies, let us provide a short set of rules that should be followed in order to construct an action free of infinitesimal copies. Just to emphasize, the method presented in this work is applicable only to exclusively $A$-dependent gauges. The rules are:

\begin{enumerate}

\item \emph{Choose and fix a gauge $\Delta^A(A)=0$ through the BRST quantization method. Call the resultant perturbative action by $S_0$.}

\item \emph{Impose BRST invariance of the gauge condition and use the gauge choice, obtain the copies equation and find the Gribov operator\footnote{The Gribov operator is the Faddeev-Popov operator with the gauge condition employed.} $\nabla^{AB}$.}

\item \emph{Define a BRST quartet formed by $(\overline{\omega},\omega,\overline{\varphi},\varphi)^{IA}$, where $I$ is a composite index describing the degeneracy of the Gribov operator.}

\item \emph{Add to $S_0$ a trivial term}

\begin{equation}
S_{triv} = s \int \overline{\omega}^{IA}\nabla^{AB}\varphi^{IB} = \int [\overline{\varphi}^{IA}\nabla^{AB}\varphi^{IB} - \overline{\omega}^{IA}\nabla^{AB}\omega^{IB} - \overline{\omega}^{IA}(s\nabla)^{AB}\varphi^{IB}] \;.
\end{equation}

\noindent \emph{So the equation of motion for $\overline{\varphi}$ reproduces the copies equation.}

\item \emph{Add a soft BRST breaking to $S_0 + S_{triv}$ of the type}

\begin{equation}
\Xi = \gamma^{2}g\int d^{4}xf^{ABC}A^{C}_{\mu}(\overline{\varphi} + \varphi)^{AB}_{\mu}\;,
\end{equation}

\noindent \emph{which will be the responsible to ruin the copies equation. The mass parameter $\gamma$ is the usual Gribov parameter. The final action given by $S_0 + S_{triv} + \Xi$ is free of infinitesimal Gribov copies}.

\end{enumerate}

With this, we have a complete set of rules to deal with infinitesimal Gribov copies, at least classically. In this work we do not present all details of the quantum aspects and we make just qualitative considerations. Two very important quantum aspects are: the \textit{renormalizability} of the model and the \textit{gap equation}. So, we can extend our set of rules to the quantum aspects with

\begin{enumerate}

\item \emph{Restore the BRST symmetry by the introduction of a suitable set of trivial sources (Zwanziger sources) so the action, which is free of copies, is embedded in a more general BRST invariant theory. With suitable physical values of the sources, the original action is recovered.}

\item \emph{Generalize all Ward identities and study the renormalizability of the action.}

\item \emph{Demand that the quantum action has minimal dependence on the mass parameter $\gamma$ and obtain the gap equation whose solution fixes this parameter.}
\end{enumerate}

\section{Yang-Mills theories and BRST symmetry}\label{GENERAL}

\subsection{Preliminary definitions}

We start formulating a generic gauge theory over a $d$-dimensional Euclidean spacetime $\mathbb{R}^d$ where $d\in\{2,3,4\}$. The gauge group is denoted by $G$, which is a semi-simple Lie group. The group elements are defined by $U=\exp\left(g\zeta\right)$ where $g$ is the coupling parameter, and the matrix $\zeta$ is an algebra-valued quantity, $\zeta=\zeta^A\Lambda^A$. Uppercase Latin indices run as $\{A,B,\ldots,H\}\in\{1,2,\ldots,\dim G\}$. The anti-hermitian matrices $\Lambda_A$ are the group generators. The respective Lie algebra is $[\Lambda^A,\Lambda^B]=f^{ABC}\Lambda^C$,
where $f^{ABC}$ are the skew-symmetric structure constants. Moreover, a Killing metric $\mathrm{Tr}(\Lambda^A\Lambda^B)=-\delta^{AB}/2$ is assumed. Besides the structure constants and the Killing metric, depending on the group, extra group invariant tensors may be necessary.

The fundamental field of a gauge theory is the algebra-valued gauge connection, $A_\mu=A^A_\mu\Lambda^A$. The connection transforms under the action of the gauge group as
\begin{equation}
A_\mu\longmapsto U^{-1}\left(\frac{1}{g}\partial_\mu+A_\mu\right)U\;.\label{gt1}
\end{equation}
At infinitesimal level, transformation \eqref{gt1} reduces to
\begin{eqnarray}
A_\mu^A\longmapsto A_\mu^A+D_\mu^{AB}\zeta^B\;,\label{gt2}
\end{eqnarray}
where the covariant derivative is defined as $D_\mu=\delta^{AB}\partial_\mu-gf^{ABC}A_\mu^C$ and $\zeta$ is the infinitesimal gauge parameter. From the covariant derivative one can easily compute the curvature, or field strength, $[D_\mu,D_\nu]^{AB}=-gf^{ABC}F^C_{\mu\nu}$, where,
\begin{equation}
F_{\mu\nu}^A=\partial_\mu A^A_\nu-\partial_\nu A^A_\mu+gf^{ABC}A^B_\mu A^C_\nu\;.\label{fs1}
\end{equation}
Thus, by demanding (i) locality, (ii) power counting renormalizability, (iii) Lorentz and (iv) gauge invariance and (v) that no parameters other than $g$ are present, the action that is encountered is the Yang-Mills action,
\begin{equation}
S_{\mathrm{YM}}=\frac{1}{4}\int d^4x F^A_{\mu\nu}F^A_{\mu\nu}\;.\label{act1}
\end{equation}

\subsection{Geometrical structure of Yang-Mills theories}

The gauge connection and field strength are mathematical structures which are formally originated on the basis of fibre bundle theory \cite{Singer:1978dk, Kobayashi, Daniel:1979ez,CottaRamusino:1985ad,Nakahara:1990th}. We define the principal bundle ${G(x)}\equiv\{G,\mathbb{R}^d\}$. The principal bundle $G(x)$ has as typical fibre, and structure group, the Lie group $G$ while the base space is, in general, the $d$-dimensional differential manifold $\mathbb{R}^d$, usually the spacetime. The total space $G(x)$ is a product between the topological spaces $G$ and $\mathbb{R}^d$. The definition of the principal bundle ${G(x)}$ is a formal manner to describe the localization of the Lie group $G$ over $\mathbb{R}^d$, assembling to each point in $\mathbb{R}^d$ a different value for the elements of $G$. The space ${G(x)}$ is then the local Lie group $G$.

The gauge connection rises on the definition of parallel transport in the total space ${G(x)}$. In fact, the gauge field components are associated with the connection 1-form $A=A_\mu dx^\mu$ while the gauge transformations are associated with coordinates changing of the total space with fixed base space coordinates.  In such transformation, which corresponds to a translation along a fixed fiber originated at $x\in\mathbb{R}^d$, the connection transforms according to \eqref{gt1}. Moreover, to every 1-form connection $A$, there is a curvature 2-form defined as $F=\mathrm{d}A+gAA=F_{\mu\nu}dx^\mu dx^\nu$, where $\mathrm{d}$ is the nilpotent exterior derivative.

The gauge connection $A$ is a natural consequence of the existence of the principal bundle $G(x)$. However, it does not belong to its former definition. The complete description of gauge theories follows from the product between ${G(x)}$ and the space of all independent connections $A_o$ that can be defined on $G(x)$, \emph{i.e.}, those connections that are not related to each other by a gauge transformation of the type \eqref{gt1}. Thus, we define the universal bundle $Y_G\equiv\{{G(x)},\mathcal{A}\}$, where the fiber and structure group are both the local Lie group ${G(x)}$, while the base space $\mathcal{A}$ is the space of all independent algebra-valued gauge connections $A_o$, the so-called moduli space. The total space $Y_G$ is then a nontrivial product between $G(x)$ and $\mathcal{A}$. The definition of a gauge orbit is obtained from \eqref{gt1} by considering a field $A_o(x)\in\mathcal{A}$ and all of its possible gauge transformations,
\begin{equation}
A=U^{-1}\left(\frac{1}{g}\mathrm{d}+A_o\right)U\;,\label{go1}
\end{equation}
which is exactly the fiber originated at $A_o(x)$. Thus, the total space $Y_G$ can be understood as the union of all gauge orbits and, hence, as the space of all possible gauge configurations that are originated in $G(x)$.

It turns out that $A$ is not the unique fundamental structure in $Y_G$. To see this, one can take the exterior derivative at the total space $Y_G$, namely $\delta$, and restrict it to the direction of a gauge orbit \eqref{go1}, $\delta|_{\mathrm{fibre}}=s$. It follows that \cite{Nakahara:1990th}
\begin{equation}
sA^{A}_\mu=-D_\mu^{AB}c^B\;,\label{brs1}
\end{equation}
where $s$ is recognized as the nilpotent BRST operator and $c=-U^{-1}sU$ is the Faddeev-Popov ghost field, which characterizes the presence of the local group as a dynamical quantity. Its BRST transformation is easily obtained,
\begin{equation}
sc^A=\frac{g}{2}f^{ABC}c^Bc^C\;.\label{brs2}
\end{equation}
The ghost field is recognized as the Maurer-Cartan 1-form in group space while \eqref{brs2} is the corresponding structure equation. The immediate consequence is the ``wrong'' statistics of the ghost field. It is a 0-form in spacetime and thus is a spin-0 field. However, it is a 1-form in group space (Because the BRST operator is an exterior derivative in group space, it increases by one the form rank in group space). The form rank in group space is known as the ghost number and must be a conserved quantity, otherwise a quantum number associated with an abstract non-physical space would be observed.

\subsection{BRST symmetry and gauge fixing}

To quantize a theory defined in $Y_G$, the introduction of a constraint is required. This constraint should select only one representative of each equivalence class described by each gauge orbit. Obviously, this is nothing else than the gauge fixing which enforces the gauge symmetry breaking so each configuration is taken into account only once at the path integral. The simplest way to do this is to define a section over $Y_G$. On the other hand, the BRST quantization method consists in study the cohomology of the BRST operator when this constraint is imposed to the Yang-Mills action \eqref{act1}. To do so, is convenient to introduce a pair of 0-form algebra-valued fields, forming a BRST doublet,
\begin{eqnarray}
s\overline{c}^A&=&ib^A\;\nonumber\\
sb^A&=&0\;,\label{brs3}
\end{eqnarray}
where $b$ is the Lautrup-Nakanishi field playing the role of a Lagrange multiplier and $\overline{c}$ is recognized as the Faddeev-Popov anti-ghost field. Ever since $b$ is a Lagrange multiplier it has vanishing ghost number, which enforces the anti-ghost to have ghost number $-1$. Thus, the ghost is now allowed to be part of the dynamics since the anti-ghost can compensate its ghost number, see Table \ref{table1}. In fact, to find the most general action depending on $A$, $c$, $\overline{c}$ and $b$ is a cohomology problem for $s$ modulo $\mathrm{d}$ and the extra usual requirements: (i) locality; (ii) Lorentz invariance and; (iii) power counting renormalizability\footnote{Global gauge invariance is also required. However, this symmetry can be broken in certain gauges, for instance, the maximal Abelian gauge.}. The solution is \cite{Piguet:1995er}
\begin{equation}
S_0=S_{\mathrm{YM}}+s\int d^4x\;\overline{c}^A\Delta_A+S_{\mathrm{ext}}\;,\label{act2}
\end{equation}
where $\Delta$ is a local functional with vanishing ghost number and dimension $\kappa$ which will define the gauge fixing constraint. In order to preserve power counting renormalizability it is demanded that $\kappa\le d$. The fields $b$ and $\overline{c}$, as any BRST doublet, remain at the trivial sector of the cohomology \cite{Piguet:1995er} which enforces $c$ to be at trivial sector as well. The last term in \eqref{act2} depends on BRST invariant external fields which are introduced to account for the non-linearity of the BRST transformations
\begin{equation}
S_{\mathrm{ext}}=s\int d^4x\left(-\Omega^A_\mu A^A_\mu+L^Ac^A\right)=\int d^4x\left(-\Omega^A_\mu D^{AB}_\mu c^B+\frac{g}{2}f^{ABC}L^Ac^Bc^C\right)\;.\label{act3}
\end{equation}
Moreover, for simplicity, we work exclusively in pure connection gauges $\Delta=\Delta(A)$, \emph{i.e.}, gauge fixing terms that contain only the gauge field. For potential gauges, the second term on the \emph{rhs} of \eqref{act2} is the gauge fixing term and reads
\begin{equation}
S_{\mathrm{gf}}=\int d^4x (ib^A\Delta_A-\overline{c}^As\Delta_A)=\int d^4x \left(ib^A\Delta_A+\overline{c}^A\frac{\delta\Delta_A}{\delta A^B_\mu}D^{BC}_\mu c^C\right)\;,\label{act4}
\end{equation}
and the gauge fixing is formally obtained by
\begin{equation}
\frac{\delta S_0}{\delta b^A}=i\Delta_A=0\;,\label{gf1}
\end{equation}

The external term allows to write the BRST operator in a functional form, which is compatible with the quantum action principle:
\begin{equation}
s=\int d^4x\left(\frac{\delta S_0}{\delta\Omega^A_\mu}\frac{\delta }{\delta A^A_\mu}+\frac{\delta S_0}{\delta L^A}\frac{\delta }{\delta c^A}+ib^A\frac{\delta}{\delta\overline{c}^A}\right)\;,\;\;\; s S_0=0\label{st1}
\end{equation}
from which we deduce the classical commutation relations, valid for any functional,
\begin{eqnarray}
\left[\frac{\delta}{\delta b^A},s\right]&=&i\frac{\delta}{\delta \overline{c}^A}\;,\nonumber\\
\left\{\frac{\delta}{\delta\overline{c}^A},s\right\}&=&0\;.\label{com1}
\end{eqnarray}

The action \eqref{act2} can now be employed in the definition of a path integral. Obviously, the renormalizability of the theory must be studied for each fixed gauge \cite{Piguet:1995er}. It turns out, however, that the gauge fixing does not ensure the complete breaking of gauge symmetry \cite{Gribov:1977wm}, a problem inherent to non-Abelian gauge theories constructed over non-trivial principal bundles $Y_G$ \cite{Singer:1978dk}. In fact, a residual gauge symmetry survives the gauge fixing process, ruining a complete quantum description for Yang-Mills theories, as briefly discussed at Sect.~\ref{INTRO}. We shall discuss in the next sections how we can control Gribov ambiguities, at least to a certain level.

\begin{table}[t]
\centering
\begin{tabular}{|c|c|c|c|c|c|c|c|c|}
\hline
fields & $A$ & $\Delta$ & $b$ & $c$ & $\bar{c}$ & $\Omega$ & $L$ & $g$ \\ \hline
Dimension & $1$ & $\kappa$ & $4-\kappa$ & $0$ & $4-\kappa$ & $3$ & $4$ & $0$\\
Ghost number & 0 & 0 & 0 & 1 & $-1$ & $-1$ & $-2$ & 0 \\ \hline
\end{tabular}
\caption{Quantum numbers of the fields. The BRST operator has ghost number $1$ and is chosen to be of dimension $0$.}
\label{table1}
\end{table}

\section{Gribov ambiguities}\label{GRIBOV}

The gauge fixing \eqref{gf1} is actually the case of the well-known gauges where Gribov ambiguity can be relatively handled, such as the Landau \cite{Gribov:1977wm,Sobreiro:2005ec,Zwanziger:1989mf,Dudal:2011gd} and the maximal Abelian \cite{Capri:2005tj,Capri:2006cz,Capri:2008ak,Capri:2008vk,Capri:2010an} gauges. Other gauges are also in this class, such as the physical gauges. A typical example of these gauges is the Coulomb \cite{Zwanziger:1998ez}, where Gribov ambiguities can be partially treated yet the renormalizability of these gauges remain unproved.

Gribov ambiguities arise when we consider two gauge configurations, $A_\mu$ and $A_\mu^\prime$, both belonging to the same equivalence class, \emph{i.e.}, they are related through a gauge transformation \eqref{gt1}. Thus, demanding that both configurations obey the same gauge fixing, one finds the usual Gribov copies equation \cite{Gribov:1977wm},
\begin{equation}
\Delta^A(A_\mu+U^{-1}D_\mu U)=0\;.\label{gribov1}
\end{equation}
For each independent configuration $A$, the solutions of \eqref{gribov1} for $U$ determine all possible redundant gauge configurations. Thus, there is still a residual gauge symmetry after gauge fixing. Interestingly, the residual gauge symmetry is a subgroup (the residual group) of the gauge group and the proof is a simple exercise in group theory.

At infinitesimal level equation \eqref{gribov1} reduces to
\begin{equation}
\nabla^A_{\phantom{A}B}\zeta^B=0\;,\label{gribov2}
\end{equation}
where
\begin{equation}
\nabla^A_{\phantom{A}B}=\frac{\partial\Delta^A(A)}{\partial A_\mu^C}D_\mu^{CB}\;,\label{gribov3}
\end{equation}
Equation \eqref{gribov2} is the infinitesimal Gribov copies equation. The infinitesimal Gribov equation can also be derived from the requirement of gauge invariance of the gauge condition,
\begin{equation}
\delta_g\Delta^A=\nabla^A_{\phantom{A}B}\zeta^B=0\;,\label{gribov4}
\end{equation}

\subsection{Three important properties}

We now provide some useful statements concerning the infinitesimal Gribov ambiguities and its relation to BRST transformations.

\begin{stmt}\label{homot}
BRST and infinitesimal gauge transformations are homotopic.
\end{stmt}

\begin{proof} The BRST operator is the exterior derivative in $Y_G$ along a typical fibre \eqref{go1}, $\delta|_{\mathrm{fibre}}=s$, thus it measures how a quantity varies along the fibre. Notwithstanding, to move along the fibre is the same as to perform a gauge transformation, since each fibre is a gauge orbit. Thus, it should be possible to identify the BRST transformations with infinitesimal gauge transformations. This can be done through a homotopy relation between both transformations.

In $Y_G$, let us denote the topological space generated by a gauge orbit at a point $x\in\mathbb{R}^d$ by $Y_x$. Then, infinitesimal gauge transformations in $Y_x$ are maps $\delta_x: Y_x\longmapsto Y_x$. Moreover, $\delta_x$ affects all other algebra-valued fields defined in $Y_x$. Thus, a point in $Y_x$ can regarded as a quartet $\Phi=(A,c,\overline{c},b)$ and a gauge transformation maps this point into another point $\Phi^\prime=(A,c,\overline{c},b)^\prime$, while maintaining $x$ fixed. The map $\delta_x$ acts at a point $\Phi$ as $\delta_x:\Phi\longmapsto\Phi+\delta_g\Phi$, where $\delta_g$ is the generator of the infinitesimal gauge transformation. Explicitly,
\begin{eqnarray}
A^A_\mu&\longmapsto&A^A_\mu+D^{AB}_\mu \zeta^B\;,\nonumber\\
c^A&\longmapsto&c^A+gf^{ABC}\zeta^Bc^C\;,\nonumber\\
\overline{c}^A&\longmapsto&\overline{c}^A+gf^{ABC}\zeta^B\overline{c}^C\;,\nonumber\\
b^A&\longmapsto&b^A+gf^{ABC}\zeta^Bb^C\;,\label{gt3}
\end{eqnarray}
where $\zeta$ is an infinitesimal gauge parameter. On the other hand, BRST transformations, at the same point $x$, are maps $s_x: Y_x\longmapsto Y_x$. Thus, $s_x:\Phi\longmapsto\Phi+\epsilon s\Phi$, where $s$ is the usual BRST operator and $\epsilon$ is a global Grassmann parameter. Explicitly, the action of $s_x$ at a point $\Phi$ is given by
\begin{eqnarray}
A^A_\mu&\longmapsto&A^A_\mu-\epsilon D^{AB}_\mu c^B\;,\nonumber\\
c^A&\longmapsto&c^A+\frac{g}{2}\epsilon f^{ABC}c^Bc^C\;,\nonumber\\
\overline{c}^A&\longmapsto&\overline{c}^A+\epsilon ib^A\;,\nonumber\\
b^A&\longmapsto&b^A\;,\label{gt3}
\end{eqnarray}

Defining the quantity $\alpha^A_t=t\zeta^A-(1-t)\epsilon c^A$, where $t\in[0,1]$, we have
\begin{eqnarray}
\delta_t&=&D^{AB}_\mu\alpha^B_t\frac{\delta}{\delta A^A_\mu}+\left[t-\frac{g}{2}(1-t)\right]f^{ABC}\alpha^B_tc^C\frac{\delta}{\delta c^A}+\left[gtf^{ABC}\zeta^B\overline{c}^C+(1-t)\epsilon ib^A\right]\frac{\delta}{\delta\overline{c}^A}+\nonumber\\
&+&gtf^{ABC}\zeta^Bb^C\frac{\delta}{\delta b^A}\;.\label{gt4}
\end{eqnarray}
Then, the map $f(t):Y_x\times[0,1]\longmapsto Y_x$, such that $\Phi\longmapsto\Phi+\delta_t\Phi$, satisfies $\delta_0=\epsilon s$ and $\delta_1=\delta_g$. Moreover, since $f$ is obviously continuous, the operator $\delta_t$ continuously deforms gauge transformations in BRST transformations and vice-versa. The proof extends to the entire space $Y_G$. Consequently, \emph{BRST and infinitesimal gauge transformations are homotopic}.
\end{proof}

It is important to keep in mind that the homotopy between BRST and infinitesimal gauge transformations has a deep geometrical meaning. In fact, in what concerns the gauge and ghost fields, this homotopy means that when we perform a BRST transformation, we are also performing a translation along a gauge orbit. With respect to the BRST doublet $(\overline{c}^A, b^A)$, the homotopy does not exactly describes a formal equivalence between a gauge transformation and a BRST transformation, only that one can be continuously deformed into another. The hypothesis of a gauge fixing that depends exclusively on $A$ plays a fundamental role here. If this was not the choice, we would have to deal with fields that do not have the same formal BRST/gauge transformations. The homotopy explained inhere may be a suggestion for a generalization of this for gauges with other fields dependence than $A_\mu^A$. However, we leave this generalization to future investigation.

An obvious consequence of this result is that, since it is valid for all infinitesimal gauge transformations, it is valid also for the residual infinitesimal gauge transformations of Gribov ambiguities. Thus, a direct consequence of the homotopy is that, if we want to eliminate the residual gauge symmetry, we have to break the BRST symmetry (see statement \ref{brs_break}). 

\begin{stmt}\label{gfp}
The operator $\nabla^A_{\phantom{A}B}$ coincides with the Faddeev-Popov operator arising from the very same gauge fixing.
\end{stmt}

\begin{proof}
This is evident from the comparison of \eqref{gribov2}, \eqref{gribov3} and \eqref{act4}. Equation \eqref{gribov2} was obtained from the gauge invariance of the gauge fixing and, since $\delta_g$ and $s$ are homotopic (statement \ref{homot}) and also that the action of $\delta_g$ and $s$ are formally the same in exclusively $A$-dependent quantities, we should also obtain \eqref{gribov2} from the requirement that the gauge fixing \eqref{gf1} is BRST invariant. To show this, we compute $s\Delta^A$ and employ the first of \eqref{com1} and the first equality in \eqref{gf1},
\begin{equation}
-\frac{\delta S_0}{\delta\overline{c}^A}=s\Delta^A\;.\label{gf0a}
\end{equation}
Thus, imposing BRST invariance of the gauge fixing, we obtain
\begin{equation}
\nabla^{AB}c_B=0\;,\label{gf0b}
\end{equation}
as we wanted to show.
\end{proof}

This small, and almost trivial, statement implies that the infinitesimal Gribov copies equation follows from the BRST variation of the gauge fixing, in accordance with \eqref{gribov2}. We also remark that expression \eqref{gf0b} is an off-shell relation obtained from the imposition of BRST invariance, we have not imposed $\Delta^A=0$. However, since Gribov copies concern configurations that obey the gauge fixing, it is natural to impose $\Delta^A=0$. As so, from now on, we will call the operator $\nabla^{AB}$ as Faddeev-Popov operator when we do not impose the gauge condition while when imposing it we will call it Gribov operator. In practice this difference should not matter because the gauge fixed path integral is supposed to take into account only configurations that obey the gauge condition.

\begin{stmt}\label{brs_break}
BRST symmetry breaking
\end{stmt}

\begin{proof}
The third property follows immediately from the second property \ref{gfp}: \emph{the absence of Faddeev-Popov zero modes requires the breaking of BRST symmetry}: This is evident from Eq.~\eqref{gf0b} (or \eqref{gribov2}), which expresses that infinitesimal Gribov copies occur when the Faddeev-Popov operator acquire zero modes. This property can be obtained from the gauge invariance of the gauge fixing or, equivalently, from BRST invariance of the gauge fixing. Thus to eliminate the residual gauge freedom, BRST symmetry must be broken.
\end{proof}

As will be discussed in the next sections, this breaking will be a soft breaking. This property will ensure the safeness of the ultraviolet sector of Yang-Mills theories.

\section{Soft BRST symmetry breaking}\label{SBRSTSB}

The homotopy between BRST and infinitesimal gauge transformations (statement \ref{homot}) allowed an identification between the infinitesimal Gribov copies and the BRST invariance of the model (statement \ref{gfp}). These two features allowed the conclusion that a BRST symmetry breaking is required to get rid of infinitesimal Gribov ambiguities (statement \ref{brs_break}). We now employ these three properties in order to eliminate the infinitesimal Gribov copies.

The BRST breaking that has to be introduced is not arbitrary, it must ensure that the Faddeev-Popov operator does not produce zero modes. The method we adopt to do this is to consider a second constraint, \emph{i.e.}, a kind of extra gauge fixing. It is an important requirement that this extra term drops out at the UV regime, otherwise it could affect the well established sector of QCD, \emph{i.e.}, the perturbative Yang-Mills action \eqref{act2} must be recovered at the UV sector. This requirement means that, at the high energy regime, the BRST symmetry must be restored \cite{Baulieu:2008fy,Baulieu:2009xr}. This is also clear from the analysis of the poles of the ghost propagator, originaly discussed by Gribov himself \cite{Gribov:1977wm} (Se also \cite{Sobreiro:2005ec,Dudal:2005na}).

A method that can be employed that ensures the ultraviolet sector is not affected is the soft BRST breaking technique \cite{Baulieu:2008fy,Baulieu:2009xr}. The method presented in this section is quite similar to that of the soft BRST symmetry breaking \cite{Baulieu:2008fy,Baulieu:2009xr}. The difference lies on the fact that, in \cite{Baulieu:2008fy,Baulieu:2009xr}, the motivation is based on the Gribov-Zwanziger action and the fitting of the propagators with the lattice data. Here, the motivation is to eliminate directly infinitesimal Gribov ambiguities in an almost\footnote{Although the method here developed can be applied to a class of gauges, the final result always depend on the original gauge constraint.} gauge fixing independent way. Moreover, the technique here developed is applied at the classical level. Thus, we can opt for a minimal change with respect to the original perturbative action \eqref{act2}. As it will become evident, this requirement is enough to reproduce the Gribov-Zwanziger BRST breaking term. On the other hand, if we require a more generalized soft term, one can achieve an improved action and a generalized gap equation will emerge (See Sect.\ref{gpeq}).

As mentioned, the idea is relatively simple. We only have to introduce an extra constraint at the Yang-Mills perturbative action \eqref{act2} in such a way that: 

\begin{cor}\label{req1}
The Faddeev-Popov operator does not develop zero modes;
\end{cor}

\begin{cor}\label{req2}
BRST symmetry is broken;
\end{cor}

\begin{cor}\label{req3}
The ultraviolet sector is not affected; so the BRST breaking is a soft breaking.
\end{cor}

\subsection{Trivial auxiliary fields}\label{trv}

The first step is to mimic the Gribov copies equation \eqref{gribov4}, which is nothing but a zero mode equation for the Faddeev-Popov operator, with the help of extra auxiliary fields. It could be easier perhaps to use the ghost field equation itself and then break the BRST symmetry by the introduction of suitable terms. However, we do not want to mess with the already well established ultraviolet sector, where the Faddeev-Popov fields play an important role on the cancelation of non-physical terms. Thus, a modification at the Faddeev-Popov sector would possibly spoil Requirement \ref{req3} as well as some of the sacred symmetries of the starting action, such as the Faddeev-Popov discrete symmetry.

The set of auxiliary fields is introduced as follows: From the fact that the infinitesimal Gribov copies equation carries $d\dim{G}$ degrees of degeneracy\footnote{We call attention to the fact that we are considering the general case where the gauge fixing is the same for all sector of the gauge group algebra. It turns out that, the generalization to gauges where different sectors of the algebra have different gauge constraints is not difficult. The example of the maximal Abelian gauge will be discussed at Sect.\ref{MAG}.} \cite{Zwanziger:1992qr} one should introduce a Lagrange multiplier with $d\dim{G}$ independent degrees of freedom. This field must then carry an independent composite index $I\equiv(\mu,B)$ describing this degeneracy and an extra group index. The extra group index allows the coupling with the Faddeev-Popov operator. The composite index runs then as $I,J...\in\{1,2,\ldots,d\dim{G}\}$. We denote this field by $\overline{\varphi}^{AB}_\mu$. Defining also its BRST doublet counterpart\footnote{Since we are running out of geometrical fields which naturally arise in $Y_G$, the only way to introduce extra fields with no reflection at the UV sector is to consider BRST doublets.}, we can write
\begin{eqnarray}
s\overline{\omega}^{AB}_\mu&=&\overline{\varphi}^{AB}_\mu\;,\nonumber\\
s\overline{\varphi}^{AB}_\mu&=&0\;,\label{brs5}
\end{eqnarray}
where $\overline{\omega}^{AB}_\mu$ plays the role of the anti-ghost while $\overline{\varphi}^{AB}_\mu$ is the Lautrup-Nakanishi analogue. Thus, choosing $\omega^{AB}_\mu$ to mimic the ghost field in \eqref{gf0b} and $\varphi^{AB}_\mu$ to be its BRST doublet pair,
\begin{eqnarray}
s\varphi^{AB}_\mu&=&\omega^{AB}_\mu\;,\nonumber\\
s\omega^{AB}_\mu&=&0\;,\label{brs6}
\end{eqnarray}
we have a BRST quartet system that, for now, should not affect the non-trivial cohomology of Yang-Mills theories \cite{Baulieu:2008fy,Baulieu:2009xr,Piguet:1995er}. We remark that the introduction of the second doublet pair \eqref{brs6} is unavoidable to maintain the Faddeev-Popov discrete symmetry. Thus, a trivial term can be introduced,
\begin{equation}
S_0\longmapsto S_\mathrm{G}=S_0+S_{triv}\;,
\end{equation} 
where
\begin{eqnarray}
S_{triv}&=& s\int d^4x\;\overline{\omega}^{AC}_\mu\nabla^{AB} \varphi^{BC}_\mu\; \nonumber\\
&=&  \int d^4x\left[\overline{\varphi}^{AC}_\mu\nabla^{AB} \varphi^{BC}_\mu-\overline{\omega}^{AC}_\mu\nabla^{AB} \omega^{BC}_\mu+\overline{\omega}^{AC}_\mu \left(D^{DE}_{\nu}c^E\right)\frac{\delta\nabla^{AB}}{\delta A_\nu^D}\varphi^{BC}_\mu\right]\;.\label{gz0}
\end{eqnarray}
From the transformations \eqref{brs5} and \eqref{brs6} it is not difficult to infer the quantum number of the extra fields, see Table \ref{table2}. This term is trivial not only because it lies on the trivial cohomology, but also because it can be eliminated from the path integral after the change of variables \cite{Zwanziger:1992qr},
\begin{equation}
\omega^{AB}_\mu\longrightarrow\omega^{AB}_\mu-\left(\nabla^{-1}\right)^{AC}(sA^{D}_{\nu})\frac{\delta\nabla^{CE}}{\delta A^D_\nu}\varphi^{EB}_\mu\;,
\end{equation}
and the fact that
\begin{equation}
\int [\overline{\varphi}\varphi\overline{\omega}\omega]\exp\left\{-\int d^4x\left(\overline{\varphi}^{AC}_\mu\nabla^{AB}\varphi^{BC}_\mu-\overline{\omega}^{AC}_\mu\nabla^{AB}\omega^{BC}_\mu\right)\right\}=1\;,\label{idexp}
\end{equation}
It is important to understand that, in expression \eqref{gz0}, the operator $\nabla$ was introduced on purpose, in such a way that the field equation for $\overline{\varphi}$ mimics the infinitesimal Gribov copies equation \eqref{gribov4}. It also characterizes the zero modes of the Faddeev-Popov operator in an independent sector with respect to the perturbative action. This is an important property because the ghost sector cannot be directly affected since we do not want to change the ultraviolet sector (requirement \ref{req3}).

\begin{table}[t]
\centering
\begin{tabular}{|c|c|c|c|c|}
\hline
fields & $\overline{\varphi}$ & $\varphi$ & $\overline{\omega}$ & $\omega$ \\ \hline
Dimension & $(4-\kappa)/2$ & $(4-\kappa)/2$ & $(4-\kappa)/2$ & $(4-\kappa)/2$\\
Ghost number & 0 & 0 & $-1$ & 1 \\ 
$Q_{d\dim{G}}$ charge & $-1$ & 1 & $-1$ & 1 \\ \hline
\end{tabular}
\caption{Quantum numbers of the auxiliary fields. The ordinary Yang-Mills fields carry vanishing $Q_{d\dim{G}}$ charge.}
\label{table2}
\end{table}

The action $S_\mathrm{G}$ displays extra symmetries that are very useful for the renormalizability problem. In fact, an extra quantum charge (see Table \ref{table2}) can be defined from
\begin{eqnarray}
Q_{IJ}S_\mathrm{G}&=&0\;,\nonumber\\
Q_{IJ}&=&\int d^4x\left(\varphi^A_I\frac{\delta}{\delta\varphi^A_J}+\omega^A_I\frac{\delta}{\delta\omega^A_J}-\overline{\varphi}^A_J\frac{\delta}{\delta\overline{\varphi}^A_I}-\overline{\omega}^A_J\frac{\delta}{\delta\overline{\omega}^A_I}\right)\;,\label{q1}
\end{eqnarray}
where the trace $Q_{II}$ characterizes a $U(d\dim{G})$ symmetry that defines the $Q_{d\dim G}$ charge. With this simplification in mind we can also define a rigid supersymmetry
\begin{eqnarray}
R_{IJ}S_\mathrm{G}&=&0\;,\nonumber\\
R_{IJ}&=&\int d^4x\left(\varphi^A_I\frac{\delta}{\delta\omega^A_J}-\overline{\omega}^A_J\frac{\delta}{\delta\overline{\varphi}^A_I}\right)\;,\label{r1}
\end{eqnarray}
and a symmetry that relates the auxiliary fields and the Faddeev-Popov ghosts
\begin{eqnarray}
T_{I}S_\mathrm{G}&=&0\;,\nonumber\\
T_{I}&=&\int d^4x\left(c^A\frac{\delta}{\delta\omega^A_I}+\overline{\omega}^A_I\frac{\delta}{\delta\overline{c}^A}\right)\;.\label{t1}
\end{eqnarray}
In a sense, symmetry \eqref{t1}, together with the Faddeev-Popov discrete symmetry, establishes an equivalence between the Faddeev-Popov ghosts and the auxiliary fields. This is a very welcome feature because the auxiliary fields must mimic the Faddeev-Popov ghost fields.

\subsection{The breaking term}\label{brk}

The next step is to include the breaking term in order to enforce that $\nabla^{AB}\varphi_{BC}\ne0$. It is immediate that the constraint term must be of the general form
\begin{equation}
S_{\mathrm{GZ}}=S_{triv}+\Xi\;,\label{gz1}
\end{equation}
where $\Xi$ is a real integrated local term, so the action can generate, at least classically, the no zero mode condition through
\begin{equation}
\frac{\delta(S_0+S_{\mathrm{GZ}})}{\delta\overline{\varphi}^{AC}_\mu}=\nabla^{AB}\varphi^{BC}_\mu+\frac{\delta\Xi}{\delta\overline{\varphi}^{AC}_\mu}=0\;.
\label{ncopies}
\end{equation}
To encounter the explicit form of $\Xi$ we demand it to obey as many symmetries of the perturbative action \eqref{act2} as possible except, obviously, the BRST symmetry, which is demanded to be violated.  Moreover, it can contain many terms. However, to break the BRST symmetry, one term is enough and we opt for a minimal change with respect to the perturbative action. We determine this term as follows:

\begin{cor}\label{req1a}
$\Xi$ must have a term which depends on $\overline{\varphi}_\mu^{AB}$.
\end{cor}

\begin{proof}[]
Obviously, without this requirement, $\Xi$ would be invisible to the operator $\delta/\delta\overline{\varphi}_\mu^{AB}$. 
\end{proof}

\begin{cor}\label{req2a}
The quantity $\delta\Xi/\delta\overline{\varphi}$ must have, among all possible terms, one term that depends exclusively on $A$.
\end{cor}

\begin{proof}[]
If this requirement is not respected, any other field appearing in this term would generate copies at their trivial vacua. Thus, in a term depending exclusively on $A_\mu^A$, the zero modes of the Gribov operator appear only at $A=0$. Nevertheless, if the trivial vacuum develops infinitesimal copies, they are necessarily different from zero and, thus, these copies will have $A_\mu^A=0$ as a copy. The consequence is that, for any vacuum copy that appears in \eqref{ncopies}, $A_\mu^A=0$ will be eliminated. Thus, to $\Xi$ be able to produce a term depending only on $A_\mu^A$, it is necessary that $\Xi$ is linear in $\overline{\varphi}_\mu^{AB}$.
\end{proof}

\begin{cor}\label{req3a}
The term $\Xi$ must be quadratic in the fields.
\end{cor}

\begin{proof}[]
With this requirement, all renormalizable field equations (characterized by Ward identities) are automatically satisfied because the possible violating terms are linear in the fields \cite{Piguet:1995er}. Then, since $\Xi$ must be linear in $\overline{\varphi}_\mu^{AB}$ (Requirement \ref{req2a}), it must also be linear in $A_\mu^A$.
\end{proof}

Thus, from these three requirements and demanding minimal modification of the action $S_\mathrm{G}$, it follows that, due to the degrees of freedom of the fields and their dimensions, the only way to construct $\Xi$ is to introduce a mass parameter\footnote{Perhaps, a higher derivative term could be considered. We, however, opt to avoid the intricacies of dealing with such terms.}. This is exactly the Requirement \ref{req3} being realised without imposing it. It follows that, the only term fulfilling these requirements is
\begin{equation}
\Xi=\int d^4x\gamma^2D_\mu^{AB}(\varphi+\overline{\varphi})^{AB}_\mu+\int d^4x\varepsilon\gamma^z\;,\label{break1}
\end{equation}
where $\gamma$ is the mass parameter with dimension given by (see Table \ref{table2})
\begin{equation}
[\gamma]=\frac{2+\kappa}{4}\;,\label{dim1}
\end{equation}
which is positive-definite for any value of $\kappa$. The pure constant term $\gamma^z$ is also allowed by power counting if, and only if, 
\begin{equation}
z=\frac{16}{(2+\kappa)}\;,\label{dim2}
\end{equation}
is an integer\footnote{By a simple algebraic analysis of \eqref{dim2}, it is easy to find that this type of term would be present if, and only if, $\kappa=2\Rightarrow n=2,\;z=4$. At this restricted class of gauges resides the Landau gauge and the maximal Abelian gauge. However, that would exclude non-local gauges (that could be localizable by a suitable set of auxiliary fields) or other gauges such as, for instance, $\partial^2\partial_\mu A_\mu=0$. In this example, these terms would be absent and the correspondent effects would probably appear from the higher derivative intricacies.}. And the combination $(\varphi+\overline{\varphi})^{AB}_\mu$ is introduced because we demand the action to be real.

For completeness, we write the full action, composed by the Yang-Mills, gauge fixing, the extra constraint and external sources terms:
\begin{eqnarray}
S&=&\frac{1}{4}\int d^4x F^A_{\mu\nu}F^A_{\mu\nu}+\int d^4x \left(ib^A\Delta_A+\overline{c}^A\frac{\delta\Delta_A}{\delta A^B_\mu}D^{BC}_\mu c^C\right)+\nonumber\\
&+&\int d^4x\left[\overline{\varphi}^{AC}_\mu\nabla^{AB} \varphi^{BC}_\mu-\overline{\omega}^{AC}_\mu\nabla^{AB} \omega^{BC}_\mu+\overline{\omega}^{AC}_\mu \left(D^{DE}_{\nu}c^E\right)\frac{\delta\nabla^{AB}}{\delta A_\nu^D}\varphi^{BC}_\mu\right]+\nonumber\\
&+&\int d^4x\gamma^2D_\mu^{AB}(\varphi+\overline{\varphi})^{AB}_\mu+\int d^4x\varepsilon\gamma^z+\int d^4x\left(-\Omega^A_\mu D^{AB}_\mu c^B+\frac{g}{2}f^{ABC}L^Ac^Bc^C\right)\;.\nonumber\\
\label{full1}
\end{eqnarray} 
and the respective $\overline{\varphi}^{AB}_\mu$ equation reads,
\begin{equation}
\nabla^{AB}\varphi^{BC}_\mu=g\gamma^2f^{ABC}A_\mu^B\;,\label{break2}
\end{equation}
which directly ensures the no zero mode condition for the entire connection functional space.

The presence of the mass parameter is very welcome because it induces a soft BRST symmetry breaking \cite{Baulieu:2008fy,Baulieu:2009xr}. In fact, the BRST variation of \eqref{full1} reads
\begin{equation}
sS=\gamma^2\int d^4x\left[-gf^{ABC}(\overline{\varphi}+\varphi)_\mu^{AB}D^{CD}_{\mu}c^D+gf^{ABC}A^{C}_{\mu}\omega_\mu^{AB}\right]\;,\label{break4}
\end{equation}
which expresses the breaking of the BRST symmetry. The fact that it is a soft breaking arises from the fact that the parameter $\gamma$ always carries non-vanishing mass dimension. This property enforces the breaking to carry a smaller dimension than the spacetime dimension. This is clear from Table \ref{table2}. The fact that the breaking is soft is of great importance to preserve standard Yang-Mills theories at the ultraviolet regime, which means that it will affect only the low energy sector of the model. It is evident then that, in the limit $\gamma\rightarrow0$, the usual perturbative action is recovered by means of \eqref{idexp}. Moreover, in this limit the BRST symmetry is also recovered. It is worth to mention that the symmetries \eqref{q1} and \eqref{r1} are no longer valid while \eqref{t1} remain being obeyed. Furthermore, the fact that \eqref{t1} is still valid is very important because it means that the symmetry between the auxiliary fields and the Faddeev-Popov ghosts is maintained. Thus, the soft breaking term ensures the absence of Faddeev-Popov zero modes. On the other hand, the BRST breaking also implies on the breaking of the $U(d\dim{G})$ and the rigid supersymmetry between the auxiliary fields. Another important remark is that the consistent introduction of the no zero mode condition and the break of the BRST symmetry enforces the introduction of a mass parameter. The BRST breaking term is proportional to this term. So, requirement \ref{req3} is, actually, redundant.

It is crucial to see that, in the soft BRST breaking method \cite{Baulieu:2008fy,Baulieu:2009xr}, the breaking is introduced by a quadratic combination of $A$ and the auxiliary fields so the propagators are modified in a soft manner. This requirement ensures renormalizability and that the ultraviolet sector is not changed, at least at the Landau gauge. Inhere, we have obtained the same type of term by requiring the elimination of Gribov infinitesimal copies through BRST breaking and this requirement naturally led to the soft BRST symmetry breaking.

\subsection{Intermission 1}\label{int1}

It has to be clear that $\nabla$, as an independent operator, will always have zero modes, but the related infinitesimal copies should not contribute to the path integral nor to the solutions of the classical field equations. When we say that the zero modes are avoided, it means that we are eliminating the copies, because these copies are responsible for the appearance of the zero modes. Moreover, the fact that the copies are eliminated solely based on $\overline{\varphi}$ field equation comes from the fact that we are adding an extra constraint to the action. Interpreting the field $\overline{\varphi}$ as a second Lagrange multiplier, the resulting equation expresses the fact that the Faddeev-Popov operator ``has no zero modes". This  equation is actually the extra constraint. Again, we remark that, when $A_\mu^A=0$, zero modes do appear. However, the null gauge configuration, if it has infinitesimal copies, is a copy of another non-vanishing configuration. Thus $A_\mu^A=0$ will be eliminated in favour of one of its copies. Another configuration that can generate ambiguity is $\varphi=0$. In this case, the constraint solution is, again, the trivial vacuum $A_\mu^A=0$, which, from the previous analysis, is harmless.

Let us analyse the classical field equations a bit more. We have, essentially, Yang-Mills theory with two constraints. Let us assume that the gauge fixing is satisfied, $\Delta=0$. Then, we have two equations involving the Faddeev-Popov operator to be analyzed, the ghost field equation $\nabla^{AB}c^B=0$ and the second constraint $\nabla^{AC}\varphi^{CB}=\gamma^2gf^{ABC}A_\mu^C$. Moreover, it is important to keep in mind that the Faddeev-Popov operator depends on the gauge field. Let us suppose, as first situation, that the second constraint is satisfied for a given non-vanishing $\varphi$. Then, if $\nabla\ne 0$, then $A\ne0$. And because the solution of this equation fixes a specific $A$ for which $\nabla\ne0$, the equation for $c$ leads to $c=0$. In contrast, as a second situation, we assume the ghost equation to be satisfied for $c\ne0$. Hence, $\nabla=0$ and the second constraint is satisfied if, and only if, $A=0$. Thus, due to the second constraint, every time that a zero mode appears, the system is thrown to the trivial vacuum configuration $A=0$.

The second situation can be interpreted as the perturbative sector of the theory because we are around the perturbative vacuum $A=0$. At this sector, $\nabla\approx0$ and the ghost field seems to be relevant. In fact, it is widely known that the ghost field is very important on the cancelation of non-physical configurations at the perturbative expansion. Under that situation, since the second constraint becomes a trivial identity ($0=0$), it is equivalent to state that $\gamma^2=0$, making the auxiliary fields actually decouple and recover the perturbative action $S_0$. If we assume that the Faddeev-Popov operator can be written as $\nabla\approx-\partial^2+0(A)\equiv p^2$ and if we associate the perturbative region with the high energy sector, then, we conclude that this is a sector where the ghosts must be massless.

On the other hand, at the first situation, there are no zero modes allowed $\nabla\ne0$, which induces $c=0$. At this region, $A\ne0$. It seems then that the ghost field attain an almost irrelevant role, an evidence that the perturbative expansion does not make sense anymore. As opposed to the second situation, this case can be associated with the non-perturbative sector because both constraints are satisfied for non-trivial configurations.

We have two distinct situations: The ultraviolet region where $S_0$ dominates, the gauge field is around its perturbative vacuum $A=0$ and the zero modes only appear as massless ghosts; and the infrared region where there are no zero modes but the theory is far from the perturbative vacuum and both constraints are simultaneously satisfied for non-trivial configurations.

\section{Beyond the classical level} \label{Qaspects}

For now, we have just employed a general method to get rid of infinitesimal Gribov ambiguities at classical level. The action \eqref{full1} is infinitesimal Gribov copies free and ready for quantization. Thus, the next step is to check its renormalizability. Although the study of renormalizability would require the specification of the gauge fixing, we can infer some universal properties of the quantum version of \eqref{full1}. In what follows, for simplicity, we restrict the model for Landau dimension gauges $\kappa=2$, which allows the presence of the term $\gamma^z$ with $z=4$ and avoids a really cumbersome dimensional analysis. The correspondent adjustments on the dimensions are displayed in Table \ref{table1a}.

\begin{table}[t]
\centering
\begin{tabular}{|c|c|c|c|c|c|c|c|c|}
\hline
fields & $\Delta$ & $b$ & $\bar{c}$ & $\overline{\varphi}$, $\varphi$, $\overline{\omega}$, $\omega$ & $\gamma$ \\ \hline
Dimension & $2$ & $2$ & $2$ & $1$ & $1$\\ \hline
\end{tabular}
\caption{Dimensions of $\kappa$-dependent fields.}
\label{table1a}
\end{table}

The main problem in studying the renormalizability of this action is the BRST breaking, which is non-linear. This means that it contaminates a perturbative expansion of the Ward identities. To deal with renormalizability there are two main methods that can be employed: i.) With the introduction of extra sources that control the BRST symmetry breaking and that, at the end, acquire certain physical values and ii.) by rewriting the Gribov-Zwanziger term in terms of a set of extra auxiliary fields that turn the BRST symmetry breaking into a linear term which remain at classical level \cite{Capri:2010hb}. Both methods have been proven to be efficient in the Landau gauge. However, for the present purposes, the first method is sufficient.

\subsection{Universal Ward identities}\label{rst}

First of all, let us write the Slavnov-Taylor operator,
\begin{equation}
\mathcal{S}(S)=\int d^4x\left(\frac{\delta S}{\delta\Omega^A_\mu}\frac{\delta S}{\delta A^A_\mu}+\frac{\delta S}{\delta L^A}\frac{\delta S}{\delta c^A}+b^A\frac{\delta S}{\delta\overline{c}^A}+\overline{\varphi}^{AB}_\mu\frac{\delta S}{\delta\overline{\omega}^{AB}_\mu}+\omega^{AB}_\mu\frac{\delta S}{\delta\varphi^{AB}_\mu}\right)\;,\label{st}
\end{equation}
which, of course, does not represent a symmetry of the action \eqref{full1},
\begin{equation}
\mathcal{S}(S)=\int d^4x\left(\frac{\delta S_{brst}}{\delta\Omega^A_\mu}\frac{\delta\Xi}{\delta A^A_\mu}+\frac{\delta S_{brst}}{\delta L^A}\frac{\delta\Xi}{\delta c^A}+b^A\frac{\delta\Xi}{\delta\overline{c}^A}+\overline{\varphi}^{AB}_\mu\frac{\delta\Xi}{\delta\overline{\omega}^{AB}_\mu}+\omega^{AB}_\mu\frac{\delta\Xi}{\delta\varphi^{AB}_\mu}\right)\ne0\;.\label{break2a}
\end{equation}
This breaking is not linear on the fields as it can be seen from expression \eqref{break1}. To control this breaking we follow \cite{Zwanziger:1992qr} by introducing a set of external sources which does not affect the cohomology of the model (See Table \ref{table4} for the respective quantum numbers)
\begin{eqnarray}
sU^{AI}_\mu&=&M^{AI}_\mu\;,\nonumber\\
sM^{AI}_\mu&=&0\;,\nonumber\\
sV^{AI}_\mu&=&N^{AI}_\mu\;,\nonumber\\
sN^{AI}_\mu&=&0\;.\label{brs8}
\end{eqnarray}
Thus, we can replace the breaking term by
\begin{equation}
\Xi\longmapsto\Xi_{brst}\;,\label{rest1}
\end{equation}
where
\begin{eqnarray}
\Xi_{brst}&=&s\int d^dx\left[-U^{AI}_\mu (D_\mu\varphi_I)^A+V^{AI}_\mu (D_\mu\overline{\omega}_I)^A+\theta U^{AI}_\mu{V}^{AI}_\mu\right]\nonumber\\
&=&\int d^{4}x\left[-M_\mu^{AI}\left( D_\mu\varphi _I\right)^A-gU_\mu^{AI}f^{ABC}\left(D_\mu c\right)^B\varphi_I^C+U_\mu^{AI}\left(D_\mu\omega_I\right)^B-N_\mu^{AI}\left(D_\mu\overline{\omega}_I\right)^A+\right.\nonumber \\
&-&\left.V_\mu^{AI}\left(D_\mu\overline{\varphi}_I\right)^A+gV_\mu^{AI}f^{ABC}\left(D_\mu c\right)^B\overline{\omega}_I^C+\theta\left(M_\mu^{AI}V_\mu^{AI}-U_\mu^{AI}N_\mu^{AI}\right)\right] \;,\label{rest2}
\end{eqnarray}
where $\theta$ is a dimensionless parameter introduced to absorb vacuum divergences. The original breaking action $\Xi$ is recovered by setting appropriate physical values for the sources. It is easy to see that the most appropriate choice is
\begin{eqnarray}
M_{\mu \nu }^{AB} &=&V_{\mu \nu }^{AB}\;\;=\;\;\gamma ^{2}\delta ^{AB}\delta _{\mu
\nu }\;,  \nonumber \\
U_{\mu \nu }^{ab} &=&N_{\mu \nu }^{ab}\;\;=\;\;0\;.\label{phys1}
\end{eqnarray}
We can see that the new parameter determine the original $\varepsilon$ parameter through
\begin{equation}
\varepsilon=\left(4\dim{G}\right)\theta\;.\label{eps}
\end{equation}

\begin{table}
  \centering
  \begin{tabular}{|c|c|c|c|c|c|c|}
\hline
Source & $U_{\mu }^{AI}$ & $M_{\mu }^{AI}$ & $N_{\mu }^{AI}$ & $V_{\mu }^{AI}$ & $K$ & $J$ 
\\ \hline\hline \textrm{dimension} & $2$ & $2$ & $2$ & $2$ & $2$ & $2$
\\ \hline $\mathrm{ghost number}$ & $-1$ & $0$ & $1$ & $0$ & $-1$ & $0$ \\ \hline
$Q_{f}\textrm{-charge}$ & $-1$ & $-1$ & $1$ & $1$ & $0$ & $0$\\
\hline
\end{tabular}
\caption{Quantum numbers of the auxiliary sources.}\label{table4}
\end{table}

\noindent The full action
\begin{equation}
\Sigma=S_{brst}+\Xi_{brst}\;,\label{fullaction}
\end{equation}
is now BRST invariant. In fact, the Slavnov-Taylor identity \eqref{st} is generalized to
\begin{eqnarray}
\mathcal{S}(S)&=&\int\left(\frac{\delta\Sigma}{\delta\Omega^A_\mu}\frac{\delta\Sigma}{\delta A^A_\mu}+\frac{\delta\Sigma}{\delta L^A}\frac{\delta\Sigma}{\delta c^A}+b^A\frac{\delta\Sigma}{\delta\overline{c}^A}+\overline{\varphi}^{AB}_\mu\frac{\delta\Sigma}{\delta\overline{\omega}^{AB}_\mu}+\omega^{AB}_\mu\frac{\delta\Sigma}{\delta\varphi^{AB}_\mu}+\right.\nonumber\\
&+&\left.M^{AI}_\mu\frac{\delta\Sigma}{\delta U^{AI}_\mu}+N^{AI}_\mu\frac{\delta\Sigma}{\delta V^{AI}_\mu}\right)=0\;.\label{st1}
\end{eqnarray}
Moreover, we also have the restoration of symmetries \eqref{q1}, \eqref{r1} and \eqref{t1}
\begin{eqnarray}
Q_{IJ}\Sigma&=&0\;,\nonumber\\
R_{IJ}\Sigma&=&0\;,\nonumber\\
T_{I}\Sigma&=&0\;,\label{syms}
\end{eqnarray}
where
\begin{eqnarray}
Q_{IJ}&=&\int d^4x\left(\varphi^A_I\frac{\delta}{\delta\varphi^A_J}+\omega^A_I\frac{\delta}{\delta\omega^A_J}-\overline{\varphi}^A_J\frac{\delta}{\delta\overline{\varphi}^A_I}-\overline{\omega}^A_J\frac{\delta}{\delta\overline{\omega}^A_I}+N^{AI}_\mu\frac{\delta\Sigma}{\delta N^{AI}_\mu}+V^{AI}_\mu\frac{\delta\Sigma}{\delta V^{AI}_\mu}+\right.\nonumber\\
&-&\left.M^{AI}_\mu\frac{\delta\Sigma}{\delta M^{AI}_\mu}-U^{AI}_\mu\frac{\delta\Sigma}{\delta U^{AI}_\mu}\right)\;,\nonumber\\
R_{IJ}&=&\int d^4x\left(\varphi^A_I\frac{\delta}{\delta\omega^A_J}-\overline{\omega}^A_J\frac{\delta}{\delta\overline{\varphi}^A_I}+V^{AI}_\mu\frac{\delta\Sigma}{\delta N^{AJ}_\mu}-U^{AJ}_\mu\frac{\delta\Sigma}{\delta M^{AI}_\mu}\right)\;,\nonumber\\
T_{I}&=&\int d^4x\left(c^A\frac{\delta}{\delta\omega^A_I}+\overline{\omega}^A_I\frac{\delta}{\delta\overline{c}^A}+U^{AI}_\mu\frac{\delta\Sigma}{\delta \Omega^A_\mu}\right)\;.\label{ops}
\end{eqnarray}

We point out that, if possible, any other symmetry of the perturbative gauge fixing should also be imposed. Equations \eqref{syms} are universal and are valid for any gauge fixing of the type $\Delta^A(A)=0$ because the difference between the case $\kappa=2$ and the others is, essentially, the pure source term. We also remark that, depending on the chosen gauge, it would be allowed to introduce interacting terms due to the Ward identities. We expect that, in that case, these terms drop out at the ultraviolet limit. We know that at the Landau and maximal Abelian gauges these interacting terms are not allowed\footnote{The maximal Abelian gauge actually requires extra interacting terms. However, these terms originate from the fact that the gauge constraint is non-linear and demand quartic ghost interactions for renormalization purposes. However, this generalized maximal Abelian gauge is a more general gauge which does not respect our restriction of gauges depending exclusively on $A_\mu^A$.}.

\subsection{The gap equation}\label{gpeq}

It remains to determine $\gamma$ in a consistent way. We have introduced the Gribov parameter in order to eliminate the zero modes of the Faddeev-Popov operator and, until now, this parameter is free. Its presence is already sufficient to eliminate the infinitesimal Gribov copies, as long as it is non-vanishing. There are several ways to determine the value of this parameter. The traditional method is demanding the quantum action $\Gamma$ to obey the renormalization group equation. However, following \cite{Zwanziger:1992qr}, one could also require that the quantum action depends minimally on the Gribov parameter. In this way, the value of $\gamma$ would be that one which minimizes the vacuum energy. This is a quite natural requirement since everything in Nature flows to the minimum energy configuration. Thus, 
\begin{equation}
\frac{\delta\Gamma}{\delta\gamma^2}=0\;.\label{gap1}
\end{equation}
This gap equation provides a non-perturbative condition to determine the explicit value of $\gamma$. Equation \eqref{gap1} is equivalent to
\begin{equation}
gf^{ABC}\left<(\varphi+\overline{\varphi})^{AB}_\mu A^{C}_{\mu}\right>=-2\varepsilon\gamma^2\label{gap2}
\end{equation}
which establishes that the Gribov parameter is associated to the condensation of the local composite operator $(\varphi+\overline{\varphi})^{AB}_\mu A^{C}_{\mu}$. In usual constructions \cite{Gribov:1977wm,Zwanziger:1992qr}, the gap equation pushes the theory very close to the Gribov first horizon. Thus, the minimization of the free energy \eqref{gap1} with respect to the Gribov parameter is responsible to ensure that the theory is driven to a non-perturbative sector. Typically, a gap equation of the form \eqref{gap1}, provides $\gamma^2\sim\exp\{-1/g^2\}$, which is a standard non-perturbative behavior.

\section{Extended formulation and alternative gap equation} \label{Gengape}

The method developed so far has reproduced the Gribov-Zwanziger approach for a generalized class of gauges. However, we have chosen a minimal alteration option with respect to the perturbative action. Any other possible dimension 2 local composite operator can be considered via the LCO technique. This improvement leads to the so called refined Gribov-Zwanziger approach \cite{Dudal:2011gd} and is very important because it produces propagators that agree with lattice predictions \cite{Cucchieri:2011ig}. 

Nevertheless, it is possible to be a bit more permissive in our construction of $\Xi$ and accommodate terms equivalent to the dimension 2 operators of the refined Gribov-Zwanziger approach. For instance, let us maintain Requirement \ref{req3a} and the fact that we are avoiding higher derivative terms. Moreover, BRST symmetry is not required for this term, only Faddeev-Popov discrete symmetry is. This notwithstanding, if we can put any term in a BRST exact form, that would be better. We demand that all extra terms depend on $\gamma$ because it ensures the ultraviolet perturbative limit to be recovered for $\gamma\rightarrow 0$. As usual, we also require that these terms are, obviously, Lorentz and colour\footnote{In certain cases, \emph{e.g.} the maximal Abelian gauge, the colour invariance has to be treated in a different way. See Sect.~\ref{MAG}} invariants. Another important condition is that terms depending on $b^A$ are excluded because it would change the gauge fixing,
\begin{equation}
\frac{\delta\widetilde{\Xi}}{\delta b^A}=0\;.
\end{equation}
We may then consider many terms coupled to the mass parameter $\gamma$ (or other mass parameter - we opt to keep only one extra parameter). Thus, it is not difficult to realize that an extra term can be added to $S$, namely
\begin{eqnarray}
\widetilde{\Xi}&=&s\int d^4x\;\gamma^2\zeta_1\overline{\omega}^{AB}_\mu\varphi^{AB}_\mu+\int d^4x\;\gamma^2\left(\zeta_2A_\mu^AA_\mu^A+\zeta_3\overline{c}^Ac^A\right)\;,\nonumber\\
&=&\int d^4x\;\gamma^2\zeta_1\left(\overline{\varphi}^{AB}_\mu\varphi^{AB}_\mu-\overline{\omega}^{AB}_\mu\omega^{AB}_\mu\right)+\int d^4x\;\gamma^2\left(\zeta_2A_\mu^AA_\mu^A+\zeta_3\overline{c}^Ac^A\right)\;.\label{break3a}
\end{eqnarray}
where the dimensionless parameters $\zeta_i$ are introduced in order to absorb eventual divergences. Eventually, these parameter may be fixed by the Ward identities, perhaps ruling (some of) them out. The gap equation \eqref{gap1} now provides a different equation for $\gamma$,
\begin{equation}
gf^{ABC}\left<(\varphi+\overline{\varphi})^{AB}_\mu A^{C}_{\mu}\right>+\zeta_1\left<\overline{\varphi}^{AB}_\mu\varphi^{AB}_\mu-
\overline{\omega}^{AB}_\mu\omega^{AB}_\mu\right>+\zeta_2\left<A_\mu^AA_\mu^A\right>+\zeta_3\left<\overline{c}^Ac^A\right>=-2\varepsilon\gamma^2\label{gap2a}
\end{equation}

We remark that, in \cite{Dudal:2011gd}, when considering the operator $\sigma=\overline{\varphi}^{AB}_\mu\varphi^{AB}_\mu-
\overline{\omega}^{AB}_\mu\omega^{AB}_\mu$ through the LCO formalism at the Landau gauge, it is found that $<\sigma>\propto\gamma^2$. Inhere, the term $\sigma$ is already proportional to $\gamma^2$, which is the reason why the gap equation \eqref{gap2} is modified to \eqref{gap2a}. Moreover, there are two possible extra terms that also contributes to the new gap equation and the solution of $\gamma$, the gluon and ghost condensates. This modification is very welcome because it prevents the theory to be thrown right at the Gribov horizon. In fact, the first Gribov horizon is characterized as the set of configurations which have infinitesimal copies. When applying the no-pole condition \cite{Gribov:1977wm,Zwanziger:1992qr,Dudal:2005na} to the ghost propagator at the thermodynamic limit, the theory is thrown right at the horizon instead of considering configurations near the horizon. The gap equation \eqref{gap2a} describes a deformation of the no-pole condition for which the operator $\left<\overline{\varphi}^{AB}_\mu\varphi^{AB}_\mu-
\overline{\omega}^{AB}_\mu\omega^{AB}_\mu\right>$ is responsible for a kind of deformation of the horizon. In the case of the refined Gribov-Zwanziger approach, this deformation is obtained by considering the condensation of $\sigma$ in an independent way of the gap equation and, eventually, it is actually found that $<\sigma>\propto\gamma^2$ (see also \cite{Vandersickel:2011zc}). We also remark that, different masses can be considered as $m^2_i=\gamma^2\zeta_i$, which may be an evidence of the equivalence between the two approaches, except for the gap equation.

Obvioulsy, the new terms can also be included into a BRST symmetric action by the introduction of extra sources which assume specific physical values, see the dimension 2 sources in Table \ref{table4}. However, renormalizability must be checked.

\subsection{Intermission 2}\label{int2}

Now, let us take a look at the path integral. In the usual approach originally developed by Gribov \cite{Gribov:1977wm}, improved by Zwanziger \cite{Zwanziger:1989mf} and refined by Dudal \emph{et al.} \cite{Dudal:2008sp}, the starting point is the fact that the Faddeev-Popov operator is hermitian, which allows a restriction in the domain of integration at the path integral. This is done by introducing a Heaviside function which turns into a delta function at a suitable thermodynamic limit. With the present approach, we perform the introduction of an extra constraint with the BRST method. Then, the improved action is quantized when exponentiation into the path integral is performed. In fact, we can go backwards from the path integral constructed with the action\footnote{See below the discussion within the alternative formulation.} \eqref{full1} and recover that what would be a generic case of the usual approaches. 

Our path integral is
\begin{equation}
Z=\int D\Phi e^{-S}\;,\label{pi1}
\end{equation}
where $D\Phi$ is the functional volume element with respect to all fields. Then, integrating the auxiliary fields $\overline{\varphi}$, $\varphi$, $\overline{\omega}$ and $\omega$, it is achieved
\begin{equation}
Z=\int DAD\overline{c}DcDb e^{-S_0-S_{GF}-S_H-\int d^4x\;\epsilon\gamma^4}\;,\label{pi2}
\end{equation}
where
\begin{equation}
S_H=\gamma^4g^2\int f^{ABC}A^C_\mu\left(\frac{1}{\nabla}\right)^{BD}f^{ADE}A^E_\mu\;.\label{hfunc1}
\end{equation}
is the horizon function equivalent. And now integrating over $b$ and the ghosts we have
\begin{equation}
Z=\int DA \det(\nabla)\delta(\Delta) e^{-S_0-S_H-\int d^4x\;\epsilon\gamma^4}\;.\label{pi3}
\end{equation}
The exponential on $\gamma$ can be eliminated by integrating over all possible values of $\gamma$. This is the inverse operation of the saddle point. It is equivalent to consider the contribution of all possible values of $\gamma$ with their respective weights. In practice, this is equivalent to assume the gap equation. Using
\begin{equation}
\int \frac{d\gamma^4}{2\pi i}e^{-\gamma^4\int d^4x\epsilon-S_H}=\delta\left(\int d^4x\epsilon+S_H\right)\;,\label{delta}
\end{equation}
we find
\begin{equation}
Z=\int DA \det(\nabla)\delta(\Delta) \delta\left(S_H+\int d^4x\epsilon\right)e^{-S_0}\;.\label{pi4}
\end{equation}
The second delta function states that the horizon function argument is fixed by the value $\epsilon$, which is finite. Thus,
\begin{equation}
f^{ABC}A^C_\mu\left(\frac{1}{\nabla}\right)^{BD}f^{ADE}A^E_\mu=-\epsilon\;.\label{hfunc2}
\end{equation}
This constraint is in complete agreement with our analysis of the classical fields equations discussed in Sect.~\ref{int1}. When $\nabla=0$, the only way to maintain $\epsilon$ is to take $A=0$ as well. This constraint actually regulates the gauge configurations and the values of the Faddeev-Popov operator in such a way that \eqref{hfunc2} is fixed through $\epsilon$ and the Faddeev-Popov operator is non-vanishing. It is interesting to see that, depending on the value of $\epsilon$, we can be very close to the analogue of the Gribov horizon. Particularly, at the Landau gauge, we are right at the horizon, which is an apparent paradox because the horizon is the place where all infinitesimal copies reside. This apparent paradox is solved with the introduction of dimension-2 condensates as demonstrated in the refined version of the Gribov-Zwanziger approach.

Another interesting point occurs when we consider the alternative formulation in the beginning of Sect.~\ref{Gengape}. For consistency\footnote{If the gauge fixing accepts the ghost mass term, even the Faddeev-Popov operator is altered. However, in this case, we cross the limits of this approach because, when restoring the BRST symmetry, this term will naturally generate a $b$-dependent term, modifying the gauge fixing itself. Thus, perhaps this term could only be considered through the LCO formalism.}, let us avoid the ghost mass term. In that case, a kind of deformation of the horizon takes place. Actually, we have a delta function which is deformed with respect to the horizon. A more formal analysis of this deformation is beyond the scope of the present work.

\section{Testing the method} \label{tests}

We now provide some tests for the model. Specifically, we apply it for the Landau and maximal Abelian gauges. In the first case, it is almost trivial to see that the method coincides with the well-known results. In the case of the maximal Abelian gauge, because the Abelian sector does not develop copies, only the non-Abelian sector requires the breaking. We restrict ourselves to the classical level and to the simpler case of \eqref{full1}, otherwise, renormalization considerations should also be taken into account.

\subsection{Landau gauge}

The Landau gauge fixing is given by
\begin{equation}
\Delta^A=\partial_\mu A_\mu^A\;,\label{landau1}
\end{equation}
which provides for the perturbative action
\begin{equation}
S_0=S_{\mathrm{YM}}+\int d^4x\left(ib^{A} \partial_\mu A_\mu^A + \overline{c}^A\partial_\mu D^{AB}_\mu c^B\right)+S_{\mathrm{ext}}\;.\label{landau2}
\end{equation}
From the BRST variation of equation (\ref{landau1}), and the imposition of BRST invariance, we can regain the Faddeev-Popov operator, which is
\begin{equation}
\nabla^{AB} = \partial_{\mu}D^{AB}_{\mu}\; . \label{landau3}
\end{equation}
Thus, from a simple substitution of \eqref{landau3} in \eqref{gz0}, we obtain the following action
\begin{equation}
S_{\mathrm{G}} =  \int d^{4}x\left[\overline{\varphi}^{AC}_{\mu}\partial_{\nu}D^{AB}_{\nu}\varphi^{BC}_{\mu} - \overline{\omega}^{AC}_{\mu}\partial_{\nu}D^{AB}_{\nu}\omega^{BC}_{\mu} + g\left(\partial_{\nu}\overline{\omega}^{AC}_{\mu}\right)f^{ABD}D^{DE}_{\nu}c^{E}\varphi^{BC}_{\mu}\right]\; , \label{landau4}
\end{equation}
The breaking term \eqref{break1} is given by
\begin{equation}
\Xi = \int d^4x \left[\gamma^{2}D_\mu^{AB}(\varphi+\overline{\varphi})^{AB}_{\mu}+\varepsilon\gamma^4\right].
\label{landau5}
\end{equation}
So, the complete classical action is written as
\begin{equation}
S = S_0 + S_{G} + \Xi\;,\label{landau6}
\end{equation}
which agrees with the well-known result for $\varepsilon=4\dim{G}$.

\subsection{Maximal Abelian gauge}\label{MAG}

In the case of the maximal Abelian gauge, we have first to set our conventions and notation and restrict ourselves, obviously, to $SU(N)$ gauge theories. In this case, the gauge field $A_{\mu}$ is decomposed into diagonal (Abelian) and off-diagonal parts as
\begin{equation}
A_{\mu} = A_{\mu}^{A}\Lambda^{A} = A_{\mu}^{a}\Lambda^{a} + A_{\mu}^{i}\Lambda^{i},
\label{adecomp}
\end{equation}
where $\Lambda^a$ are the off-diagonal sector of generators and $\Lambda^i$ the Abelian generators. Since all generators $\Lambda^{i}$ commute with each other, they generate the Cartan subgroup of $SU(N)$. To avoid confusion, we have to keep in mind that capital indices $\left\{A, B, C,\ldots \right\}$ are related to the entire $SU(N)$ group, and so, they run in the set $\left\{1, \ldots, (N^{2} - 1)\right\}$. Small indices $\left\{a, b, c, \ldots h\right\}$ represent the off-diagonal part of $SU(N)$ and they vary in the set $\left\{1, \ldots, N(N-1)\right\}$. Finally, small indices $\left\{i,j,k, \ldots\right\}$ describe the Abelian part of $SU(N)$ and they run in the range $\left\{1, \ldots, (N-1)\right\}$. From the $SU(N)$ Lie algebra, we can write the following decomposed algebra
\begin{eqnarray}
\left[\Lambda^{a},\Lambda^{b}\right] &=&f^{abc}\Lambda^{c} + f^{abi}\Lambda^{i}, \nonumber \\
\left[\Lambda^{a},\Lambda^{i}\right] &=& -f^{abi}\Lambda^{b}, \nonumber \\
\left[\Lambda^{i},\Lambda^{j}\right] &=& 0.
\label{algmag}
\end{eqnarray}
Using Jacobi identity, we can show that the structure constants satisfy
\begin{eqnarray}
f^{abi}f^{bcj} + f^{abj}f^{bic} &=& 0, \nonumber \\
f^{abc}f^{cdi} + f^{adc}f^{cib} + f^{aic}f^{cbd} &=& 0, \nonumber \\
f^{abc}f^{cde} + f^{abi}f^{ide} + f^{adc}f^{ceb} + f^{adi}f^{ieb} + f^{aec}f^{cbd} + f^{aei}f^{ibd} &=& 0.
\label{jacobimag}
\end{eqnarray}
Proceeding in this way, we can write the off-diagonal and diagonal components of a infinitesimal gauge transformation with parameter $\alpha$ respectively
\begin{eqnarray}
\delta_{g}A_{\mu}^{a} &=& -(D_{\mu}^{ab}\alpha^{b} + gf^{abc}A_{\mu}^{b}\alpha^{c} + gf^{abi}A_{\mu}^{b}\alpha^{i}), \nonumber \\
\delta_{g}A_{\mu}^{i} &=& - (\partial_{\mu}\alpha^{i} + gA_{\mu}^{a}\alpha^{b}),
\label{gaugetransmag}
\end{eqnarray}
where the covariant derivative $D_{\mu}^{ab}$ is defined with respect to the abelian component of the gauge field, \textit{i.e.}
\begin{equation}
D_{\mu}^{ab} = \delta^{ab}\partial_{\mu} - gf^{abi}A_{\mu}^{i}.
\label{cderivative}
\end{equation}
Following the same decomposition, the Yang-Mils action reads
\begin{equation}
S_{\mathrm{YM}} = \frac{1}{4}\int d^{d}x(F_{\mu \nu}^{a}F_{\mu \nu}^{a} + F_{\mu \nu}^{i}F_{\mu \nu}^{i}),
\label{ymactionmag}
\end{equation}
with
\begin{eqnarray}
F_{\mu \nu}^{a} &=& D_{\mu}^{ab}A_{\nu}^{b} - D_{\nu}^{ab}A_{\mu}^{b} + gf^{abc}A_{\mu}^{b}A_{\nu}^{c}, \nonumber \\
F_{\mu \nu}^{i} &=& \partial_{\mu}A_{\nu}^{i} - \partial_{\nu}A_{\mu}^{i} + gf^{abi}A_{\mu}^{a}A_{\nu}^{b}.
\label{fmunu}
\end{eqnarray}

The decomposed BRST transformations are easily obtained from \eqref{brs1}, \eqref{brs2} and \eqref{brs3}. For the off-diagonal fields we have,
\begin{eqnarray}
sA_{\mu}^{a} &=& -(D_{\mu}^{ab}c^{b} + gf^{abc}A_{\mu}^{b}c^{c} + gf^{abi}A_{\mu}^{b}c^{i}), \nonumber \\
sc^{a} &=& gf^{abi}c^{b}c^{i} + \frac{g}{2}f^{abc}c^{b}c^{c}, \nonumber \\
s\overline{c}^{a} &=& ib^{a}, \nonumber \\
sb^{a} &=& 0,
\label{brstoff}
\end{eqnarray}
and for the Abelian sector,
\begin{eqnarray}
sA_{\mu}^{i} &=& -(\partial_{\mu}c^{i} + gf^{abi}A_{\mu}^{a}c^{b})\;, \nonumber \\
sc^{i} &=& \frac{g}{2}f^{abi}c^{a}c^{b}\;, \nonumber \\
s\overline{c}^{i} &=& ib^{i}\;, \nonumber \\
sb^{i} &=& 0\;.
\label{brstabelian}
\end{eqnarray}

Now, we introduce the gauge conditions that characterize the maximal Abelian gauge. The maximal Abelian gauge condition is obtained by fixing the non-Abelian sector in a Cartan subgroup covariant way,
\begin{equation}
D_{\mu}^{ab}A_{\mu}^{b}=0\;.\label{magcond}
\end{equation}
This condition maintain the Abelian gauge symmetry, which is usually fixed by the Landau condition:
\begin{equation}
\partial_{\mu}A_{\mu}^i=0\;.
\label{magcond2}
\end{equation}

By performing the BRST variations of conditions \eqref{magcond} and \eqref{magcond2} and imposing their invariance we obtain the Gribov copies equations
\begin{eqnarray}
\nabla^{ab}c^b&=&0\;,\nonumber\\
\partial_\mu\left(\partial_\mu c^i+gf^{abi}A_\mu^a c^b\right)&=&0\;,\label{gribovmag1}
\end{eqnarray}
where
\begin{equation}
\nabla^{ab} = D_{\mu}^{ac}D_{\mu}^{cb} + gf^{acd}A_{\mu}^{c}D_{\mu}^{db} + g^{2}f^{aci}f^{bdi}A_{\mu}^{c}A_{\mu}^{d}.
\label{fpmag}
\end{equation}
is recognized as the Gribov operator in the maximal Abelian gauge. The second of \eqref{gribovmag1} is actually redundant because it can be easily solved once the first equation is solved \cite{Capri:2010an}. Thus, if we ruin the first of \eqref{gribovmag1}, the second is automatically spoiled.  Let us start with the $SU(2)$ case.

\subsubsection{The $SU(2)$ case}

In the case of $SU(2)$, there is only one Abelian generator and it will be denoted with no group index, for instance, for the Abelian gluon field $A_{\mu}$. For the off-diagonal components we will write $A^{a}_{\mu}$, with $a=1,2$ and the covariant derivative is defined by $D^{ab}_{\mu} = \partial_{\mu} \delta^{ab} - g\epsilon^{ab}A_{\mu}$, where $ f^{ab3}\equiv\epsilon^{ab3}\equiv\epsilon^{ab}$. The perturbative action is given by
\begin{equation}
S_0 = S_{\mathrm{YM}} + \int d^dx(ib^{a}D^{ab}_{\mu}A^{b}_{\mu} + \overline{c}^{a}\nabla^{ab}c^{b} + g\epsilon^{ab}\overline{c}^{a}cD^{bc}_{\mu}A^{c}_{\mu} + ib\partial_{\mu}A_{\mu} + \overline{c}\partial_{\mu}(\partial_{\mu}c+g\epsilon^{ab}A^{a}_{\mu}c^{b})) + S_{\mathrm{ext}},
\label{magpert}
\end{equation}
where the Faddeev-Popov operator \eqref{fpmag} is written as
\begin{equation}
\nabla^{ab} = D^{ac}_{\mu}D^{cb}_{\mu} + g^{2}\epsilon^{ac}\epsilon^{bd}A^{c}_{\mu}A^{d}_{\mu},
\label{fpsu2}
\end{equation}
It is important to highlight that this operator does not have any abelian component because of the redundancy on the second equation in \eqref{gribovmag1}. Thus, the trivial extra action \eqref{gz0} is considered only for the operator \eqref{fpsu2}. This means that the auxiliary fields carry an off-diagonal index to contract with $\nabla^{ab}$ and the composite index $(a,\mu)$ to account for the degeneracy of equations \eqref{gribovmag1}. Thus, \eqref{gz0} turns into
\begin{equation}
S_{triv} =  \int d^4x \left[ \overline{\varphi}^{ac}_{\mu}\nabla^{ab}\varphi^{bc}_{\mu} - \overline{\omega}^{ac}_{\mu}\nabla^{ab}\omega^{bc}_{\mu} - \overline{\omega}^{ac}_{\mu}(sA^{d}_{\nu})\frac{\delta \nabla^{ab}}{\delta A^{d}_{\nu}}\varphi^{bc}_{\mu} - \overline{\omega}^{ac}_{\mu}(sA_{\nu})\frac{\delta \nabla^{ab}}{\delta A_{\nu}}\varphi^{bc}_{\mu}\right],
\label{sgsu2mag}
\end{equation}
which, by performing the appropriate calculations, provides
\begin{eqnarray}
S_{triv} &=&  \int d^4x \left\{\overline{\varphi}^{ac}_{\mu}\nabla^{ab}\varphi^{bc}_{\mu} - \overline{\omega}^{ac}_{\mu}\nabla^{ab}\omega^{bc}_{\mu} -\overline{\omega}^{ac}_{\mu}\left[2g\epsilon^{ad}(\partial_{\nu}c + g\epsilon^{ef}A^{e}_{\nu}c^{f})D^{db}_{\nu} + \right.\right.\nonumber\\ 
&+& \left.\left.g\epsilon^{ab}\partial_{\nu}(\partial_{\nu}+g\epsilon^{ef}A^{e}_{\nu}c^{f})-g^{2}(\epsilon^{ad}\epsilon^{be} + \epsilon^{ae}\epsilon^{bd})(D^{dg}_{\nu}c^{g} + g\epsilon^{dg}A^{g}_{\nu}c)A^{e}_{\nu}\right]\varphi^{bc}_{\mu}\right\}\;.\label{su2}
\end{eqnarray}

The breaking term is constructed as in Sect.~\ref{brk}. The difference is that the auxiliary fields are exclusively non-Abelian. Thus, to couple the bosonic auxiliary fields with the gauge field, we have only one possibility which is to couple them with the Abelian sector,
\begin{equation}
\Xi = \int d^4x \left[g\gamma^{2}D^{ab}_\mu(\varphi+\overline{\varphi})^{ab}_{\mu}+\varepsilon \gamma^4\right],
\label{breakmag}
\end{equation}
and the complete action is consistent with the known result \cite{Capri:2005tj}. In the same way of the Landau gauge, the parameter $\varepsilon$ must be fixed through renormalization considerations.

\subsubsection{The $SU(N)$ case}

The procedure for this case is completely analogous to the one made for the $SU(2)$ case. However, now, instead of one Abelian generator we have $N-1$ Abelian generators. The perturbative action is given by \cite{Capri:2010an}
\begin{eqnarray}
S_0 &=& S_{\mathrm{YM}} + \int d^4x \left[ib^{a}D^{ab}_{\mu}A^{b}_{\mu} + \overline{c}^{a}\nabla^{ab}c^{b} - gf^{abc}(D^{ad}_{\mu}A^{d}_{\mu})\overline{c}^{b}c^{c} \right.\nonumber \\
&-& \left.gf^{abi}(D^{ac}_{\mu}A^{c}_{\mu})\overline{c}^{b}c^{i} + ib^{i}\partial_{\mu}A^{i}_{\mu} + \overline{c}^{i}\partial_{\mu}(\partial_{\mu}c^{i} + gf^{abi}A^{a}_{\mu}c^{b})\right] + S_{\mathrm{ext}}
\label{pertmag}
\end{eqnarray}
The trivial term is constructed in the same way of the $SU(2)$ case, except that the Faddeev-Popov operator is the general one \eqref{fpmag},

\begin{eqnarray}
S_{triv} &=& \int d^4x \left\{\overline{\varphi}^{ac}_{\mu}\nabla^{ab}\varphi^{bc}_{\mu} - \overline{\omega}^{ac}_{\mu}\nabla^{ab}\omega^{bc}_{\mu} - \overline{\omega}^{ac}_{\mu}\left[2gf^{aek}(\partial_{\nu}c^{k} + gf^{ghk}A^{g}_{\nu}c^{h})D^{eb}_{\nu}\right.\right.\nonumber \\
&+& \left.\left.gf^{abk}\partial_{\nu}(\partial_{\nu}c^{k} + gf^{ghk}A^{g}_{\nu}c^{h}) + g^{2}f^{aef}f^{fbk}A^{e}_{\nu}(\partial_{\nu}c^{k}+gf^{ghk}A^{g}_{\nu}c^{h}) \right.\right.\nonumber \\
&-& \left.\left.gf^{adf}(D^{dg}_{\nu}c^{g} + gf^{dgh}A^{g}_{\nu}c^{h} - gf^{dgj}A^{g}_{\nu}c^{j})D^{fb}_{\nu} \right.\right.\nonumber \\
&-& \left.\left.g^{2}(f^{adi}f^{bei} + f^{aei}f^{bdi})A^{e}_{\nu}(D^{dg}c^{g} + gf^{dgh}A^{g}_{\nu}c^{h} + gf^{dgj}A^{g}_{\nu}c^{j}) \right]\varphi^{bc}_{\mu}\right\},
\label{sgmag}
\end{eqnarray}

The breaking term can also be determined as in Sect.~\ref{brk}. Besides a term of the form $D(\varphi+\overline{\varphi})$ as in \eqref{breakmag}, we can also introduce a term like $f^{abc}A^a(\varphi+\overline{\varphi})^{bc}$, which is absent in the $SU(2)$ case. It follows that
\begin{equation}
\Xi = \int d^4x \left[\gamma^2\left(D_\mu^{ab}+\xi gf^{abc}A^{c}_{\mu}\right)(\varphi+\overline{\varphi})^{ab}_{\mu}+\varepsilon \gamma^4\right].
\label{ximag}
\end{equation}
The new parameter $\xi$ is dimensionless and is introduced to account for possible divergences. Once again, the complete action coincides with the known result \cite{Capri:2010an}. Moreover, it is easy to verify that expressions \eqref{sgmag} and \eqref{ximag} reduce to \eqref{su2} and \eqref{breakmag} for $N=2$.

\section{Conclusions}

We have developed a new method that eliminates infinitesimal Gribov ambiguities. In this method, infinitesimal Gribov ambiguities are eliminated through the introduction of an extra constraint. Strictly speaking, what are eliminated are the gauge configurations that leads the Faddeev-Popov operator to develop zero modes. It was shown that, this constraint must eliminate the Faddeev-Popov operator zero modes and, simultaneously, break the BRST symmetry in a soft manner. The method generalizes the known approaches to treat Gribov ambiguities at the Landau and maximal Abelian gauges. However, the method can be sistematically generalized to other gauges, as long as they depend exclusively on the gauge field. We have also shown that some of the Ward identities that rises from the auxiliary fields are universal, in the sense that they can always be defined at the class of gauges we are considering. Finally, as a consistency check, we have applied the method for the Landau and maximal Abelian gauges. The results are compatible with the standard results encountered in the literature. Moreover, the method developed is compatible with the soft breaking of BRST symmetry technique. In this method, the motivation is to achieve the confining propagators of the refined Gribov-Zwanziger approach at the Landau gauge. Here, we have directly demanded the elimination of the copies and the resulting action is the Gribov-Zwanziger action. 

A few words must be said about the method developed: our goal was to ruin the copies equation directly, which means that we wanted to avoid the zero modes of the Faddeev-Popov operator. Since we do not want to modify the anti-ghost equation of motion, we introduced a set of auxiliary fields and a BRST breaking term in order to impose a new constraint into the theory. In this way, the ultraviolet sector remains untouched. This constraint is exactly the equation of motion for $\bar{\varphi}$ which can be visualized as an eigenvalue equation for the Faddeev-Popov operator. The zero eingenvalues appear only at $A=0$, which is a ``safe point'' at the gauge configuration space, \emph{i.e.}, it is a region which is identified with the perturbative sector of the theory. The $\overline{\varphi}$ equation must be faced as the new constraint; the equation that eliminates the configurations where the Faddeev-Popov operator develops zero modes. The extra constraint induces only gauge configurations for which the Faddeev-Popov operator does not have zero modes. Hence, only these configurations will appear at the ghost sector as well. Finally, the action we obtain is the same obtained in the usual approaches for the Landau and maximal Abelian gauges, a fact that stands in favour of the method. In both cases, the ghost propagator will be positive-definite, \emph{i.e.}, the Faddeev-Popov operator is positive-definite. This is exactly the restriction to the first Gribov region and means that this restriction can be viewed as consequence of our method when applied to these gauges. In the case of non-hermitian Faddeev-Popov operators it is very difficult to say something about a geometric region at the functional gauge space. However, a similar delta function is found (see \eqref{hfunc2}) stating that the horizon function analogue is also fixed to a ``region'' fixed by the constant $\epsilon$.

It is worth mention that, the Gribov-Zwanziger approach, or its refined version, depends on the fact that the Faddeev-Popov operator is hermitian. Although we have not applied the method to a different gauge other than the Landau and maximal Abelian gauges, our method applies to gauges where the Faddeev-Popov operator is non-hermitian. The reason is that we have not demanded hermiticity of the Faddeev-Popov operator in any way. We left for future investigation the employment of the technique to gauges where the Faddeev-Popov operator is non-hermitian. Specifically, it is already under investigation \cite{work1} one particular interesting gauge which interpolates between the Landau and maximal Abelian gauges \cite{Dudal:2005zr,Capri:2005zj,Capri:2006bj}. As important issues to be investigated, we have the generalization of the technique to gauges that also depend on other fields such as the Curci-Ferrari gauge \cite{Curci:1976kh}, the renormalizable maximal Abelian gauge and others and the study of the effects that the generalized gap equation presented here could bring to the theory.

\section*{Acknowledgements}

The authors are grateful to M.~A.~L.~Capri, M.~S.~Guimar\~aes, D.~Dudal, S.~A.~Dias and L.~Bonora for very useful discussions. The Conselho Nacional de Desenvolvimento Cient\'{i}fico e Tecnol\'{o}gico\footnote{ RFS is a level PQ-2 researcher under the program \emph{Produtividade em Pesquisa}, 308845/2012-9.} (CNPq-Brazil), The Coordena\c c\~ao de Aperfei\c coamento de Pessoal de N\'ivel Superior (CAPES) and the Pr\'o-Reitoria de Pesquisa, P\'os-Gradua\c c\~ao e Inova\c c\~ao (PROPPI-UFF) are acknowledge for financial support.

\end{document}